\pdfoutput=1

\documentclass[sigconf]{aamas} 
\usepackage[inline]{enumitem}

\def\Rat{\mathbb{Q}}                
\newcommand{\set}[1]{\{#1\}}
\usepackage{balance} 
\usepackage{tikz}
\newtheorem{remark}{Remark}
\usepackage[section]{placeins}
\usepackage{thm-restate}
\usepackage{algorithm}
\usepackage{algpseudocode}
\usepackage{bm}

\newcommand\blfootnote[1]{%
  \begingroup
  \renewcommand\thefootnote{}\footnotetext{#1}%
  \endgroup
}

\newif\ifcomments

\newif\ifarXiv

\arXivtrue

\newif\ifshortauthors

\acmSubmissionID{774}

\title[AAMAS-2023 Formatting Instructions]{$k$-Prize Weighted Voting Games}

\ifshortauthors
    \author{Wei-Chen Lee$^{*1}$, David Hyland$^{*1}$, Alessandro Abate$^{1}$, Edith Elkind$^{1}$, Jiarui Gan$^{1}$, Julian Gutierrez$^{2}$, Paul Harrenstein$^{1}$, Michael Wooldridge$^{1}$}
    \affiliation{
      \institution{$^1$University of Oxford, Oxford, UK; $^2$Monash University, Melbourne, Australia}
      }
    \email{{wei-chen.lee, david.hyland, aabate, elkind, jiarui.gan, paul.harrenstein, mjw}@cs.ox.ac.uk}
    \email{julian.gutierrez@monash.edu}
\else
    \author{Wei-Chen Lee*}
    \affiliation{
      \institution{University of Oxford}
      \city{Oxford}
      \country{United Kingdom}}
    \email{wei-chen.lee@cs.ox.ac.uk}
    
    \author{David Hyland*}
    \affiliation{
      \institution{University of Oxford}
      \city{Oxford}
      \country{United Kingdom}}
    \email{david.hyland@cs.ox.ac.uk}
    
    \author{Alessandro Abate}
    \affiliation{
      \institution{University of Oxford}
      \city{Oxford}
      \country{United Kingdom}}
    \email{alessandro.abate@cs.ox.ac.uk}
    
    \author{Edith Elkind}
    \affiliation{
      \institution{University of Oxford}
      \city{Oxford}
      \country{United Kingdom}}
    \email{edith.elkind@cs.ox.ac.uk}
    
    \author{Jiarui Gan}
    \affiliation{
      \institution{University of Oxford}
      \city{Oxford}
      \country{United Kingdom}}
    \email{jiarui.gan@cs.ox.ac.uk}
    
    \author{Julian Gutierrez}
    \affiliation{
      \institution{Monash University}
      \city{Melbourne}
      \country{Australia}}
    \email{julian.gutierrez@monash.edu}
    
    \author{Paul Harrenstein}
    \affiliation{
      \institution{University of Oxford}
      \city{Oxford}
      \country{United Kingdom}}
    \email{paul.harrenstein@cs.ox.ac.uk}
    
    \author{Michael Wooldridge}
    \affiliation{
      \institution{University of Oxford}
      \city{Oxford}
      \country{United Kingdom}}
    \email{michael.wooldridge@cs.ox.ac.uk}
\fi

\begin{abstract}
We introduce a natural variant of weighted voting games, which we refer to as $k$-Prize Weighted Voting Games. Such games consist of $n$ players with weights, and $k$ prizes, of possibly differing values. The players form coalitions, and the $i$-th largest coalition (by the sum of weights of its members) wins the $i$-th largest prize, which is then shared among its members. We present four solution concepts to analyse the games in this class, and characterise the existence of stable outcomes in games with three players and two prizes, and in games with uniform prizes. We then explore the efficiency of stable outcomes in terms of Pareto optimality and utilitarian social welfare. Finally, we study the computational complexity of finding stable outcomes.
\end{abstract}


\newcommand{\BibTeX}{\rm B\kern-.05em{\sc i\kern-.025em b}\kern-.08em\TeX}

\begin{document}


\pagestyle{fancy}
\fancyhead{}


\maketitle 


\section{Introduction} \label{sec:intro}

Weighted voting games (WVGs) are an elegant and practically important class of simple cooperative games, which have been extensively studied in the multi-agent systems community~\cite{elkind:2009,chalkiadakis2011computational}\blfootnote{$^*$Equal contribution.}.
In a WVG, each player has a numeric weight, and the game is played by the players forming coalitions. The weight of a coalition is defined as the sum of the weights of its members, and a coalition's value is $1$, if the coalition's weight exceeds a predefined threshold $q > 0$, and is $0$ otherwise. In this way, a coalition is said to be \textit{winning} if the total weight of all its members exceeds the static threshold $q$. WVGs are usually designed so that there can be just one winning coalition: simple majority games, where each player has a weight of $1$, and a coalition is winning if it contains a majority of the players, are an obvious example.

However, there are settings that can be modelled by WVGs with more than one winning coalition. For example, political parties may form coalitions after an election, which then compete to form the Government (1st prize) or the Opposition (2nd prize). This observation motivates us to consider a variation of WVGs in which there is more than one prize, and the criterion for winning a prize is determined by the weight of a coalition \textit{relative} to others. We refer to these new games as \emph{$k$-Prize Weighted Voting Games}. The idea of such games is very intuitive: players form coalitions, and the coalition with the largest cumulative weight takes first prize, the coalition with the second largest weight takes second prize, and so on, until all $k$ prizes have been allocated or all coalitions have been considered. Note that all coalitions that do not rank among the $k$ heaviest ones, if any, receive a zero payoff. 

It is important to note that our model permits {\em externalities}: unlike in conventional WVGs, the performance of a coalition in a $k$-prize game depends not just on its own makeup, but also on the other coalitions that form. Formally, $k$-Prize Weighted Voting Games belong to a class of games known as \textit{partition function games} (PFGs), where the value of a coalition depends both on the coalition's members as well as the coalitions that other players form  \cite{thrall1962generalized, thrall1963n, zhao1992hybrid}. 

The remainder of this paper is structured as follows: In Section~\ref{sec:preliminaries}, we formally define the game and an outcome of the game. We then define four deviation concepts along with their corresponding stability concepts in Section~\ref{sec:deviation_concepts}. In Section~\ref{sec:characterising_stability} we provide necessary and sufficient conditions for the existence of stable outcomes in two sub-classes of games: games
with three players and two prizes, and games with uniform prizes. In Section~\ref{sec:efficiency}, we explore the efficiency of stable outcomes in terms of Pareto optimality and utilitarian social welfare. Section~\ref{sec:complexity} focuses on the computational complexity of stability-related problems, such as finding a beneficial deviation, deciding stability of a given outcome in a game, and determining the existence of stable outcomes in a given game. Finally, we survey related work in Section~\ref{sec:related_work} and conclude in Section~\ref{sec:conclusions}.

\section{Preliminaries} \label{sec:preliminaries}

A \textit{$k $-Prize Weighted Voting Game} consists of $n$ players with weights, and $k$ prizes, where $k \leq n$, of possibly differing positive rational values. Players form coalitions, and the $i$-th 
largest coalition (by the sum of weights of its members) wins the $i$-th largest prize, which 
is shared among its members. If there are more than $k$ coalitions, then coalitions that are 
not among the $k$ largest coalitions do not receive a prize. A tie-break relation $\succ$ 
determines which coalition wins a given prize in the case where multiple coalitions have the 
same weight. Formally, the game is defined as follows:

\begin{definition} [$k$-WVG]
    A \textit{$k$-Prize Weighted Voting Game} ($k$-WVG or also simply \textit{game}) is a tuple $G = (N, \vec{w},  \vec{p}, \succ)$, where 
    \begin{itemize}
        \item $N = \{1, \ldots ,n\}$ is the set of $n$ \textit{players};
        \item $\vec{w} = (w_1, \ldots , w_n) \in \Rat_+^n$ is a vector of \textit{weights}, where $w_i$ is the weight of player $i$. We use $w(S) := \sum_{i \in S} w_i$ to denote the sum of weights of a set of players $S \subseteq N$;
        \item $\vec{p} = (p_1, \ldots , p_k) \in \Rat_+^k$ is a vector of positive \textit{prize values} in descending order, where $k \leq n$.
        \item $\succ$ is a strict total order on the powerset of all players $2^N$, where given two sets $A, B \subseteq N$, $A \succ B$ means that $w(A) > w(B)$, or $w(A) = w(B)$ and $A$ wins the tie-break against $B$.
    \end{itemize}
    Subsuming~$N$ and~$\succ$, we also denote  $G=(N, \vec{w},  \vec{p}, \succ)$ by $(\vec w;\vec p)$.
\end{definition}

Players form a partition $\pi \in \Pi(N)$ of $N$, where $\Pi(X)$ represents the set of all partitions of a finite set $X$ into pairwise disjoint subsets. The \textit{value} $v(C, \pi)$ of a coalition $C$ within a partition $\pi$ is the value of the prize that is won by $C$. Formally, we have
    \[v(C, \pi) =
            \begin{cases}
		        p_i & \text{if $C$ is the $i^{\text{th}}$ largest coalition in $\pi$ w.r.t. $\succ$}\\
                0 & \text{otherwise}.
	        \end{cases}
    \] 

When we refer to the largest or the smallest player/coalition, we implicitly use $\succ$ to break ties. In all computational complexity issues that we consider, we assume access to a polynomial time oracle that can decide, given two coalitions $C_1$ and $C_2$, whether $C_1 \succ C_2$ or $C_2 \succ C_1$. We say that a coalition $C$ in a coalition structure $\pi$ is \textit{winning} if it wins one of the $k$ prizes, i.e., $v(C,\pi) >  0$.

In the remainder of the paper, we refer to $k$-WVGs as simply games, and we use the notation $(\vec{w}; \vec{p})$ to represent a game, where it is assumed that $N = \{1, \ldots, |\vec{w}|\}$, and that some $\succ$ is given.

\begin{definition} [Outcome]
    An {\em outcome} of a $k$-WVG is a tuple $\omega = (\pi, \vec{x})$, where: 
    \begin{itemize}
        \item $\pi = \{C_1, \ldots , C_m\}$ is a partition of $N$ into coalitions.
        \item $\vec{x} = (x_1, \ldots, x_n)$ is the vector of payoffs, where $x_i$ is the payoff to player $i$, such that $x_i \geq 0$ for all $i \in N$ and $\sum_{i \in C_j} x_i = v(C_j, \pi)$ for all $j \in \set{1,\ldots,m}$. Given a subset of players $S \subseteq N$, we write $x(S)$ to denote the sum of the payoffs to these players, i.e., $x(S) := \sum_{i \in S} x_i$. 
    \end{itemize}
\end{definition}

We write $\Omega$ for the set of all possible outcomes of a game. Where there is no 
ambiguity, we may use the shorthand `$i\ldots k$' to present the coalition $\{i, \ldots, k\}$ 
in a partition, and `$|$' to separate coalitions within the partition, enclosed by square 
brackets. For example, the partition of three players $\{\{1, 2\}, \{3\}\}$ can be written 
as $[1\ 2\ |\ 3]$. Conventionally, we order the coalitions within a partition according to $\succ$,
so we can easily identify the coalition winning a certain prize.

\begin{example} [A $k$-WVG and an outcome] \label{eg:game_and_outcome}
    Consider a game of 4 players and 2 prizes: $G = (\vec{w}; \vec{p})$ where $\vec{w} = (0.45, 0.25, 0.24, 0.06)$ and $\vec{p} = (2, 1)$, and an outcome $\omega = (\pi, \vec{x})$ where $\pi = [2\ 3\ |\ 1\ |\ 4]$ and $\vec{x} = (1, 0.8, 1.2, 0)$. In this outcome, players 2 and 3 form a coalition to take the first prize: $v(\{2, 3\}, [2\ 3\ |\ 1\ |\ 4]) = p_1 = 2$. They split the prize as $(x_2, x_3) = (0.8, 1.2)$. Player 1 forms a singleton coalition, i.e., a coalition consisting of only one player, and claims the second prize of value $v(\{1\}, [2\ 3\ |\ 1\ |\ 4]) = p_2 = 1$ by itself. Player 4 is left without a prize and receives a payoff of $0$.
\end{example}

Note that the outcome in Example \ref{eg:game_and_outcome} is intuitively `unstable' for a 
number of reasons: First, player 1 can form a coalition with player 2 or player 4 (or both), 
so as to be guaranteed to win first prize ($p_1 = 2$), which is greater than the sum of 
their current payoffs ($x_1 + x_2 = x_1 + x_2 + x_4 = 1.8$ and $x_1 + x_4 = 1$). Similarly, 
players 2 and 4 can together secure second prize ($p_2 = 1$), which is greater than the 
sum of their current payoffs ($x_2 + x_4 = 0.8$). Or perhaps player 2 could `go at it alone' and 
try to secure second prize ($p_2 = 1 > x_2 = 0.8$), depending on whether players~3 and~4 
would form a coalition and also compete for the second prize? To reason about stability, we 
must first formalise the concept of a deviation.

\section{Deviation and Stability Concepts} \label{sec:deviation_concepts}

We introduce four deviation concepts and their corresponding notions of stability. These concepts differ by what the deviating players are permitted to do: whether they form a single coalition, and, in the case that they split up into multiple deviating coalitions, the incentive requirements that must be satisfied. In all our deviation concepts, the deviating players are \textit{maximally pessimistic} regarding the actions of the \textit{residual} (non-deviating) players, i.e., they assume that the residual players would partition themselves in order to prevent the deviation from being successful if possible.

\begin{definition} [SCD]
    A \textit{single-coalition deviation} (SCD) from an outcome $(\pi,\vec{x})$ is a set $C \subseteq N$ such that for all $\pi' \in \Pi(N)$ with $C \in \pi'$, we have
    \[
        \sum_{i \in C} x_i < v(C, \pi').
    \]
\end{definition}

\begin{definition} [MCD]
    A \textit{multi-coalition deviation} (MCD) from an outcome $(\pi,\vec{x})$ is a partition $\pi_D \in \Pi(D)$ of a set $D \subseteq N$ such that for all $\pi' \in \Pi(N)$ with $\pi_D \subseteq \pi'$ and all $C \in \pi_D$, we have
    \[
        \sum_{i \in C} x_i < v(C, \pi').
    \]
\end{definition}

\begin{definition} [wMCD]
    A \textit{weak multi-coalition deviation} (wMCD) from an outcome $(\pi, \vec{x})$ is a partition $\pi_D \in \Pi(D)$ of a set $D \subseteq N$ such that for all $\pi' \in \Pi(N)$ with $\pi_D \subseteq \pi'$, it holds that:
    \begin{enumerate}
        \item for all $C \in \pi_D$ we have
        \[
        \sum_{i \in C} x_i \leq v(C, \pi');
        \]
        \item for some $C \in \pi_D$, the inequality in $(1)$ is strict.
    \end{enumerate}
\end{definition}

\begin{definition} [MCDT]
    A \textit{multi-coalition deviation with inter-coalition transfer} (MCDT) from an outcome $(\pi, \vec{x})$ is a partition $\pi_D \in \Pi(D)$ of a set $D \subseteq N$ such that for all $\pi' \in \Pi(N)$ with $\pi_D \subseteq \pi'$, we have
    \[
        \sum_{i \in D} x_i < \sum_{C \in \pi_D} v(C, \pi').
    \]
\end{definition}

For a deviation concept $\theta \in \{SCD, MCD, wMCD, MCDT\}$, we say that an outcome $(\pi,\vec{x})$ is \textit{susceptible to a} $\theta$ if there exists a $\theta$ from $(\pi,\vec{x})$. With these concepts defined, we can introduce the notion of stability with respect to each deviation concept.

\begin{definition} [Stability and Core] 
    Given any of the above deviation concepts $\theta \in \{SCD, MCD, wMCD, MCDT\}$, let $\theta(G) \subseteq \Omega$ denote the set of outcomes that are susceptible to a $\theta$ in the game $G$. We say an outcome is $\theta$\textit{-stable} if it is not susceptible to a $\theta$. The set of outcomes that are $\theta$-stable is the $\theta$\textit{-core} of the game, i.e., $\theta\textit{-core(G)} = \Omega \setminus \theta(G)$. 
\end{definition}

It is worth noting that the concept of the SCD-core coincides with that of \textit{$\alpha$-core} \cite{AumannPeleg1960} in the more general partition form games.

\begin{remark} [Irrelevance of coalition structure] \label{rem:irrelevance_of_coalition_structure}
    For a given outcome $(\pi, \vec{x})$, observe that no consideration is given to the coalition structure $\pi$ when assessing whether it is susceptible to some deviation concept; Only the payoff vector $\vec{x}$ is relevant. It follows that two outcomes with identical payoff vectors must be susceptible to the same deviation concepts, and have the same stability properties. Formally, given any two outcomes $\omega = (\pi, \vec{x})$, $\omega' = (\pi', \vec{x}')$ where $\vec{x} = \vec{x}'$, $\omega$ is stable to some deviation concept if and only if $\omega'$ is stable to the same deviation concept. We can therefore also say whether a payoff vector $\vec{x}$ is susceptible (or stable) to some deviation concept.
\end{remark}

Now, we analyse the relationship between the different deviation concepts. The following relations are immediate:

\begin{remark} [Hierarchy of solution concepts] \label{rem:hierarchy}
    For every game $k$-WVG~$G$ it holds that 
    \[SCD(G) \subseteq MCD(G) \subseteq wMCD(G) \subseteq MCDT(G).\] 
    Consequently, we have
    \[\textit{MCDT-core}(G) \subseteq \textit{wMCD-core}(G) \subseteq \textit{MCD-core}(G) \subseteq \textit{SCD-core}(G).\] 
\end{remark}

In summary, in an SCD, all deviating players form a single deviating coalition; an MCD allows the deviating players to form multiple deviating coalitions (i.e., a deviating partition) so that each deviating coalition must be strictly better off; a wMCD is similar to an MCD, except only one deviating coalition needs to be strictly better off, while others are weakly better off; an MCDT requires only that the total value of all deviating coalitions exceeds the current payoff of all deviating players, with the implicit assumption that coalitions can make transfers to one another to ensure that all deviating players are strictly better off. 

Intuitively, we can think of an MCD as multiple deviating coalitions `colluding' with each other, so that each deviating coalition is better off and does not try to prevent the other deviators from succeeding (as the residual players are assumed to do). A wMCD is, in a sense, a `hopeful' MCD, where the deviating coalitions that are strictly better off (the `MCD' component) rely on the other deviating coalitions to be content with their payoff (i.e., $v(C, \pi') = x(C)$) and not act against them. In some cases, there is another interpretation that does not rely on the good will of some deviating coalitions. If, depending on the residual players' response, each of the deviating coalitions has some chance of being strictly better off, then each of the deviating coalitions has the incentive to form a wMCD. 

We now show that SCD, MCD, wMCD and MCDT are, in fact, distinct concepts.\footnote{Note that Proposition \ref{SCD_neq_MCD} depends on a specifically chosen tie-breaking order.}

\begin{proposition} \label{SCD_neq_MCD}
    For some $k$-WVG $G$, $SCD(G) \subsetneq MCD(G)$.
\end{proposition}

\begin{proof}
    By definition, $SCD(G) \subseteq MCD(G)$ for each $G$. We show that there exists a game $G$ and an outcome $\omega$ such that $\omega \in MCD(G)$ and $\omega \not \in SCD(G)$, through the use of tie-breaking. 
    
    Consider a game with six players $N = \{1,\ldots, 6\}$ all with weight~$1$ and two prizes $\vec{p} = (3, 2)$. Assume a tie-breaking order such that, for any two distinct coalitions $A$ and $B$ both containing exactly two members, we have $A \succ B \succ \{3, 4\}$ whenever $6 \in B \setminus A$ and $A \neq \{3,4\}$, i.e., $A$ beats $B$ if $B$ contains $6$ and $A$ does not, unless $A = \{3, 4\}$.

    Consider the outcome $\omega = (\pi,\vec x)$ with $\pi=[1\ 2\ 3\ |\ 4\ 5\ |\ 6]$ and $\vec x=(1,1,1,1,1,0)$.
    Let $X = \{1, 2\}$ and $Y = \{5, 6\}$. We find that $[X|Y]$ constitutes an MCD. To see this, note that $x(X) = 2 < p_1$ and $x(Y) = 1 < p_2$, that is, $X$ would prefer $p_1$ and $Y$ would prefer $p_2$ compared to their respective payoffs in $\omega$. Observe that $N \setminus (X \cup Y) = \{3, 4\}$. Accordingly, $X \succ Y \succ \{3,4\}$, so players $3$ and $4$ cannot prevent this deviation by either staying separate or forming a coalition. Thus, $\omega$ is susceptible to an MCD by~$[X\ | Y]$.
    
    However, $\omega$ is not susceptible to an SCD. To see this, first observe that every 
    coalition $A$ that would prefer $p_1$ to its current payoff 
    has either (1) $|A| = 3$ and $6 \in A$, or (2) $|A| \le 2$ (as otherwise $x(A)\geq 3$). 
    Let $B = N \setminus A$. In case (1) we have $|B| = 3$ and $6 \not \in B$. Hence, $B \succ A$. 
    In case (2) we have $w(B) \ge 4 > 2 = w(A)$. In either case, $A$ cannot be guaranteed 
    to win first prize. On the other hand, we note that for every coalition $A$ that would prefer 
    $p_2$ to its current payoff, i.e., $x(A) < p_2$, either (1) $|A| = 1$ or (2) $|A| = 2$, 
    $6 \in A$. In case (1), observe that $|N \setminus A| = 5$ and we can partition 
    $N \setminus A$ into sets $B$ and $B'$ such that $|B|=3$ and $|B'|=2$. Then, 
    $B \succ B' \succ A$. In case (2), we can partition $N \setminus A$ into sets $B$ and $B'$, 
    where $|B|=|B'|=2$ and $B, B'\neq \{3, 4\}$. Thus, we have $B \succ B' \succ A$. 
    In either case, we find that $A$ cannot be guaranteed to win second prize. 
    The outcome~$\omega$ is therefore not susceptible to an SCD.
\end{proof}

\begin{proposition} \label{MCD_neq_wMCD}
    For some $k$-WVG $G$, $MCD(G) \subsetneq wMCD(G)$.
\end{proposition}

\begin{proof}
    Consider the game $G = ((0.4,0.3,0.3),(2,1))$ and the outcome $\omega = ([2\ 3\ \vert\ 1],(1,1,1))$. It can be verified that $\omega$ is MCD-stable but not wMCD-stable (see Proposition \ref{prop:3_player_2_prize_game} case 2).
\end{proof}

\begin{proposition} \label{wMCD_neq_MCDT}
    For some $k$-WVG $G$, $wMCD(G) \subsetneq MCDT(G)$.
\end{proposition}

\begin{proof}
    Since $wMCD(G) \subseteq MCDT(G)$ by definition, we show that there exists a game $G$ and an outcome $\omega$ such that $\omega \in MCDT(G)$ and $\omega \not \in wMCD(G)$.
    
    Consider the 4-player, 3-prize game $G = ((7, 4, 4, 2); (2, 1.1, 1.1))$ and the outcome $\omega = ([1\ |\ 2\ |\ 3\ |\ 4], (2, 1.1, 1.1, 0))$. We show that $\omega$ is susceptible to an MCDT. Consider the deviating partition $[2\ 3\ |\ 4]$: the only coalition structure that can be induced by the deviating partition is $([2\ 3\ |\ 1\ |\ 4])$, where the deviating coalition would earn the 1st and 3rd prize ($2 + 1.1 = 3.1$), exceeding their current payoff ($x(\set{2,3,4}) = 2.2$). Note that this particular deviation is not a wMCD, since players 2 and 3 would have no incentive to deviate ($x(\set{2,3}) = 2.2 > 2 = p_1$) without some transfer from player~4.
    
    Suppose for contradiction that $\omega$ is also susceptible to a wMCD, for some deviating coalition $D \subseteq N$. It cannot be the case that $|D| = 4$ (i.e., everyone jointly deviates), because all prizes are claimed, and hence no deviation can lead to a strict improvement for some deviating coalition and no loss to others. Also, it cannot be the case that $|D| = 1$, since each player already forms a singleton coalition.
    
    Consider the case $|D| = 2$: $D$ cannot form two singleton deviating coalitions, since they could do no better than the existing outcome; For $D$ to form a single deviating coalition of two members, the coalition must be able to achieve a prize that is strictly better than its current payoff. Clearly, any pair involving player 1 can do no better (player 1 already wins the largest prize), nor does any pair involving player 4, since the other player can do no better. Finally, the pair $(2, 3)$, who can form a coalition larger than $\set{1}$, would not benefit by doing so, as the first prize is smaller than their current joint payoffs under $\omega$.
    
    Finally, consider the case $|D| = 3$: The deviating players cannot form a partition of the form $[i\ j\ k]$ because there is no single prize that exceeds any three players' combined payoff, nor can they form $[i\ |\ j\ |\ k]$ because such a deviation would result in the same outcome. We are left with the case of $[i\ j\ |\ k]$. From the case $|D| = 2$, we saw that no pair $[i\ j]$ could benefit from a deviation, and yet no singleton $[k]$ could strictly improve their payoff. Therefore, there are no feasible deviations with $|D| = 3$. Hence, $\omega$ is not susceptible to a wMCD and is susceptible to an MCDT.
\end{proof}

\section{Characterising stability} \label{sec:characterising_stability}

Before we explore the computational complexity of finding stable outcomes in general $k$-WVGs, we show that for two classes of $k$-WVGs this is relatively straightforward.

\vspace*{1ex}\noindent\textbf{3-player 2-prize games:}
We first consider the special class of games with 3-players and 2-prizes and show that we can easily decide whether such games have a stable outcome, and identify one such outcome. Then, without loss of generality, we normalise the prize values in such games to $(R, 1)$, where $R > 1$.

\begin{restatable}{proposition}{Threeplayertwoprizegame} \label{prop:3_player_2_prize_game}
    A $k$-WVG with 3 players and 2 prizes, where players are labelled in descending order according to $\succ$, has an SCD/ MCD-stable outcome if and only if $\{1\} \succ \{2, 3\}$, or $\{2, 3\} \succ \{1\}$ and $R \leq 2$. In the first case, an SCD/ MCD-stable outcome must have payoff $\vec{x} = (R, 1, 0)$; in the second case an SCD/ MCD-stable outcome must have payoff $\vec{x} = (1, 1, R-1)$.
    
    A game has a wMCD/ MCDT-stable outcome if and only if $\{1\} \succ \{2, 3\}$, and a wMCD/ MCDT-stable outcome must have payoff $\vec{x} = (R, 1, 0)$.
\end{restatable}

\begin{proof}[Sketch of proof]
    We provide an outline of the proof, consisting of all the cases that need to be considered.\footnote{The full proof is provided \ifarXiv in Appendix A \else in appendix A of [citation] \fi} 
    
    Firstly, observe that there are only two tie-breaking orders consistent with $\set 1\succ\set 2\succ\set 3$, namely, one with $\set 1\succ\set{2,3}$ and one with $\set{2,3}\succ \set 1$.
	
	The first case we consider is when $\set 1\succ\set{2,3}$. For this, it suffices to show that: (i) $\vec{x} = (R, 1, 0)$ is MCDT-stable and (ii) $\vec{x} = (R, 1, 0)$ is the only SCD-stable payoff vector. The next case to consider is when $\set{2,3}\succ \set 1$ and $1 < R \leq 2$. In this case, we can establish that: (i) $\vec{x} = (1, 1, R-1)$ is MCD-stable; (ii) $\vec{x} = (1, 1, R-1)$ is not wMCD-stable; and (iii) $\vec{x} = (1, 1, R-1)$ is the only SCD-stable payoff. Finally, the last case is when $\set{2,3}\succ \set 1$ and $R > 2$. In this scenario, a simple contradiction argument suffices to show that no payoff vector $\vec{x}$ is SCD-stable.

	Exploiting Remark~\ref{rem:hierarchy}, these findings give us our result, which is also summarised in Table~\ref{tab:3_player_2_prize}.    
\end{proof}

\begin{table}[t]
    \centering
    \caption{Payoff vector in the core of 3-player 2-prize games for each case. $\emptyset$ signifies an empty core.}
    \begin{tabular}{r@{\,$\succ$\,}lcc}
         \toprule
        \multicolumn{2}{c}{Case}           & SCD/ MCD-core         & wMCD/ MCDT-core \\
         \midrule
         $\{1\}$ & $\{2, 3\}$                 & $(R, 1, 0)$           & $(R, 1, 0)$  \\
         $\{2, 3\}$ & $\{1\}, R \leq 2$       & $(1, 1, R-1)$         & $\emptyset$  \\
         $\{2, 3\}$ & $\{1\}, R > 2$          & $\emptyset$           & $\emptyset$  \\
         \bottomrule
    \end{tabular}
    \label{tab:3_player_2_prize}
\end{table}

\begin{remark} [Emptiness of the core]
    The results above (summarised in Table \ref{tab:3_player_2_prize}) illustrate that the core of a game, in relation to any deviation concept, can be empty.
\end{remark}

\vspace*{1ex}\noindent\textbf{Uniform-prize games:} 
We now consider a special class of $k$-WVG where all prizes are of equal value (normalised to 1), which we refer to as \textit{$k$-uniform-prize WVGs}. We show that for this special class of games, deciding whether a game has an SCD-stable outcome can be easily done via a recursive argument. The following definitions for uniform-prize games establish the necessary concepts for our main result.

\begin{definition} [Singleton winner]
    A player $i$ is a \textit{singleton winner} of a $k$-uniform-prize WVG if it can win a prize as a singleton coalition, regardless of the coalition structure of the remaining players, i.e., for all $\pi \in \{\pi' \in \Pi(N) \mid \{i\} \in \pi'\}$, we have $v(\{i\}, \pi) = 1$.
\end{definition}

\begin{lemma} \label{lem:singleton_winners}
    For any $k$-uniform-prize WVG, if $n = k$, then all players are singleton winners; if $n > k$, a player $i \in N$ is a singleton winner if and only if $w_i > w^*$ where $w^* = \frac{1}{k + 1} w(N)$, or $w_i = w^*$ and for every partition $\pi$ of the players $N \setminus \{i\}$ into $k$ coalitions, there is a coalition $C_j \in \pi$ such that $\{i\} \succ C_j$.
\end{lemma}

\begin{proof}
    If $n = k$, then at most $k$ coalitions can be formed by all the players, and therefore each player is guaranteed a prize. If $n > k$, then a player $i$ can secure one of the $k$ prizes as a singleton if $w_i > w^*$, because the collective weight of the remaining players is $w(N \setminus \{i\}) = w(N) - w_i < w(N) \frac{k}{k+1} = k w^*$ and cannot form $k$ coalitions each of size at least $w^*$. If instead $w_i = w^*$, then $w(N \setminus \{i\}) = k w^*$ and it may be possible for the remaining players to form exactly $k$ coalitions each of size $w^*$. In this case, the singleton player $i$ can win one of the $k$ prizes if and only if for all partitions of $N \setminus \{i\}$ into $k$ coalitions of weight $w^*$, the singleton beats (on tie-break) at least one of the $k$ coalitions in the partition.
\end{proof}

\begin{definition} [Reduction \& Irreducible game]
Let $s$ be the set of singleton winners of a $k$-uniform-prize WVG $G$. A \textit{reduction} of $G$ is the removal of all singleton winners, and the removal of the same number of prizes as there are singleton winners, from the game. If $k > |s|$, the reduction of a game $G$ results in a uniform-prize game $G'$ with $n-|s|$ players and $k-|s|$ prizes, with the remaining players having the same weight as in $G$; if $k = |s|$, then no prizes are left after the reduction and there is no reduced game. A game with no singleton winner is \textit{irreducible}.
\end{definition}

\begin{corollary} \label{cor:irreducible_game}
    Let $w^* = \frac{1}{k + 1} w(N)$. In an irreducible game, either i) each player has weight $< w^*$, or ii) one player $i$ has weight $w_i = w^*$, all remaining players each have weight $< w^*$, and the remaining players can form $k$ coalitions each of weight $w^*$ and beat $\{i\}$ on tie-break. 
\end{corollary}

\begin{definition} [ISW]
    An \textit{iterative singleton winner} (ISW) of a game $G$ is one of the singleton winners removed in the iterative reduction of $G$ until the game cannot be further reduced.
\end{definition}

Inspired by the concept of maximin share as it occurs in the literature on fair allocations (see, e.g., \cite{budish2011MMS}), we now define the notion of an MMS partition and an MMS set.

\begin{definition} [MMS partition \& MMS set]
    A partition $\pi \in \Pi(N)$ is a \textit{maximin share (MMS) partition} of players $N$ into $l$ sets if it maximises the weight of the smallest set with respect to $\succ$. Because $\succ$ is a strict total order, an \textit{MMS partition} always exists, and all MMS partitions have the same smallest set, which we refer to as the \textit{MMS set} or $MMS(N, l)$.
\end{definition}

Observe that for any set of players $S$ within an MMS partition that is not the smallest (i.e., $S \in \pi \setminus MMS(N, l)$), we have $MMS(N, l) = MMS(N\setminus S, l-1)$.

\begin{lemma} \label{lem:irreducible_game_empty_SCD_core}
    The SCD-core of an irreducible game is empty.
\end{lemma}

\begin{proof}
    Let $G$ be an irreducible game with $n$ players and $k$ prizes, and let the `winning threshold' $w^* = \frac{1}{k + 1} w(N)$. We show that we can find an SCD in any outcome $\omega = (\pi, \vec{x})$. By Remark \ref{rem:irrelevance_of_coalition_structure}, we need only to consider the payoff vector $\vec{x}$ of an outcome and not its coalition structure $\pi$.
    
    Consider an MMS partition $\phi$ of $N$ into $k + 1$ sets. Label the $k + 1$ sets of players as $S_1, \dots, S_{k+1}$ in descending order, i.e., $S_a \succ S_b$ for all $a < b$, and $S_{k+1} = MMS(N, k+1)$. Note that $S_{k+1}$ cannot be empty since $n \geq k+1$ in an irreducible game. We now construct an SCD by considering each $S_d$ where $d \in \set{1,\ldots,k}$ (i.e. excluding the smallest set $S_{k+1}$). 
    
    Suppose there exists $d$ such that $x(S_d) < 1$, then $S_d$ can form an SCD. To see this, note that $S_{k+1} = MMS(N, k+1) = MMS(N \setminus~S_d, k)$, and that $S_d \succ S_{k+1}$. This means that under any partition of $N \setminus S_d$ into $k$ sets, $S_d$ will always beat the smallest set and win one of the $k$ prizes. $S_d$ can therefore guarantee a prize value of $1 > x(S_d)$ by deviating.
    
    If there does not exist $d$ such that $x(S_d) < 1$, then it must be the case that $x(S_d) = 1$ for all $d \in \set{1,\ldots,k}$ and $x(S_{k+1}) = 0$, because the total prize value to be shared amongst all players is exactly $k$. We consider two cases of the distribution of $\vec{x}$. First is the case where $x_i \in \{0, 1\}$ for all $i \in N$, i.e., the payoff to each player is either $1$ or $0$. This can occur either because all winning coalitions are singletons, or because some winning coalitions contain all but one zero-payoff players. Since there are $k$ prizes, there are exactly $k$ players with payoff $1$, and by Corollary \ref{cor:irreducible_game}, either a) each winning player has weight less than $w^*$, or b) at most one has weight $w^*$ but loses on all tie-breaks. If $k > 1$, the total weight of all zero-payoff players must exceeds $w(N) - k w^* = w^*$, and these players can guarantee to win a prize by forming an SCD. If $k = 1$ and the winning player has weight $w^*$, then by Corollary \ref{cor:irreducible_game} the zero-payoff players can also form an SCD of size $w^*$ and take the prize from the winning player.
    
    Now consider the case where $x_i \in (0,1)$ for some $i \in N$, i.e., there is some set $S_d$ in which at least 2 players have a positive payoff. Let the smallest of these players be player $i$ with weight $w_i$. We argue that $S_d \cup S_{k+1} \setminus \{i\}$ can form an SCD. First note that because $x_i > 0$, $S_d \setminus \{i\}$ can incentivise $S_{k+1}$ to join the SCD and improve its own payoffs by sharing the gain from removing $i$. Moreover, we show that the residual players cannot prevent this deviating coalition from winning a prize by considering two cases: $\{i\} \succ S_{k+1}$ and $\{i\} \prec S_{k+1}$. 
    
    If $\{i\} \succ S_{k+1}$, then we can swap $i$ and $S_{k+1}$ in an MMS partition and increase the weight of the smallest coalition in the partition (since both $S_d \cup S_{k+1} \setminus \{i\} \succ S_{k+1}$ and $\{i\} \succ S_{k+1}$). This contradicts the premise that $MMS(N, k+1) = S_{k+1}$.
    
    If $\{i\} \prec S_{k+1}$, then we show that $S_d \cup S_{k+1} \succ S_{k+1} = MMS(N, k+1) \succ MMS(N \setminus (S_d \cup S_{k+1}), k)$: if the last ordering were false, i.e., $MMS(N \setminus (S_d \cup S_{k+1}), k) \succ MMS(N, k+1)$, then we can add $S_d \cup S_{k+1}$ as a set to the MMS partition of $N \setminus (S_d \cup S_{k+1})$ and achieve a better MMS set than $MMS(N, k+1)$, resulting in a contradiction. Since $S_d \cup S_{k+1} \succ MMS(N \setminus (S_d \cup S_{k+1}), k)$, $S_d \cup S_{k+1}$ can guarantee to win one of the $k$ prizes by forming an SCD.
    
    Having considered all cases, we have shown that there is no SCD-stable outcome in an irreducible game.
\end{proof}

\begin{theorem} [Stability of uniform-prize games]
    \label{thm:uniform-prize_stability}
    A $k$-uniform-prize WVG $G$ has an SCD-stable outcome (i.e., a non-empty SCD-core) if and only if there are $k$ iterative singleton winners. Moreover, an SCD-stable outcome has the unique payoff $\vec{x} = (x_1, \ldots x_n)$ such that 
    \[x_i = 
    \begin{cases}
        1 \quad \text{if $i$ is one of the $k$ iterative singleton winners}, \\
        0 \quad \text{otherwise.}
    \end{cases}
    \]
\end{theorem}

\begin{proof}
    For the forward direction, suppose for a contradiction that there exists a stable outcome in a game with $\ell < k$ ISWs. Then, we can iteratively reduce the game until we are left with an irreducible game $G'$. By Lemma \ref{lem:irreducible_game_empty_SCD_core}, the SCD-core of the irreducible game is empty. This means that every outcome in $G'$ is susceptible to an SCD by some coalition $C$. Now, let $\omega = (\pi,\vec{x})$ be an outcome in $G$. If any of the $\ell$ ISWs in $G$ do not receive a payoff of $1$, then they can form a singleton SCD. If all of the singleton winners receive a payoff of $1$, then because some coalition $C$ has an SCD to secure some prize $p_i \in \set{p_{\ell+1},\ldots,p_k}$ in the irreducible game $G'$, then they also have an SCD to secure $p_i$ in $G$, because the $\ell$ ISWs can at best only take the first $\ell$ prizes.
    
    For the backward direction, suppose that there are $k$ ISWs. Then, by definition, these ISWs can be guaranteed to win a prize on their own, and no deviation by the players who are not ISWs can secure one of these prizes. Moreover, since there are only $k$ prizes, each ISW wins a prize and does not stand to gain by deviating. Hence, all outcomes where the $k$ ISWs receive a payoff of $1$ are stable.
\end{proof}

Although this characterisation is specific to the SCD-core, it follows from Remark \ref{rem:hierarchy} that the MCD/ wMCD/ MCDT-core of the game is also empty. Moreover, if the MCD/ wMCD/ MCDT-core is empty, then it cannot be the case that $G$ contains $k$ ISWs (because the outcome where the $k$ ISWs each win a prize is stable to all deviation concepts).

\begin{corollary}\label{corollary:uniform_prize_cores}
    The MCD/ wMCD/ MCDT-core of a $k$-uniform-prize game is non-empty if and only if it has $k$ iterative singleton winners.
\end{corollary}

Finally, we show that in fact, in $k$-uniform-prize games, the cores under all solution concepts coincide.

\begin{corollary}
    In a $k$-uniform-prize game $G$, SCD-core(G) = MCD-core(G) = wMCD-core(G) = MCDT-core(G).
\end{corollary}

\begin{proof}
    If $G$ does not consist of $k$ ISWs, then by Corollary \ref{corollary:uniform_prize_cores} the core, with respect to any deviation concept, is empty. If $G$ consists of $k$ ISWs, then either each of the $k$ ISWs gets payoff $1$ (i.e., does not share its prize with others), or some ISWs get payoff $<1$. In the first case, we have shown that no deviation exists. In the second case, an ISW with a payoff of $<1$ can deviate by forming an SCD on its own. Thus an outcome is either not stable to any deviation concept or is stable to all deviation concepts.
\end{proof}

Theorem 4.9 also provides us with a polynomial time algorithm for deciding whether the core of a $k$-uniform-prize WVG is empty, provided that it is easy to break ties.

\begin{corollary}
    The problem of deciding whether a $k$-uniform-prize WVG has an SCD-stable outcome is in $P$.
\end{corollary}

\begin{proof}
    This follows directly from Theorem \ref{thm:uniform-prize_stability}: we can sort the players from largest to smallest and identify whether there are $k$ ISWs in polynomial time by iterative reduction, which will occur at most $k$ times.
\end{proof}

\section{Efficiency} \label{sec:efficiency}

A natural consideration about stable outcomes is whether they are socially efficient. We consider two well-known concepts of efficiency: Pareto-optimality and utilitarian social welfare. An outcome is Pareto-optimal if there is no alternative outcome where some players can gain more without some other player losing; it is utilitarian if the total payoff for all players cannot be further increased. Observe that a utilitarian outcome is always Pareto-optimal (but not the reverse); and that, in a $k$-WVG, an outcome is utilitarian if and only if all prizes are claimed. Formally, an outcome $\omega = (\pi, \vec{x})$ is \textit{Pareto-optimal} if there does not exist some $\omega' = (\pi', \vec{x}')$ such that there exists some $i \in N$ where $x'_i > x_i$, and for all $i \in N$, we have $x'_i \geq x_i$. An outcome $\omega = (\pi, \vec{x})$ is \textit{utilitarian} if there is no outcome $\omega' = (\pi', \vec{x}')$ such that $\sum_{i \in N} x'_i > \sum_{i \in N} x_i$. A utilitarian outcome is always Pareto-optimal.

We begin with a negative result about the efficiency of SCD/ MCD-stable outcomes, followed by two positive results about wMCD/ MCDT-stable outcome.

\begin{proposition}
    There are SCD/MCD-stable outcomes that are neither Pareto-optimal nor utilitarian.
\end{proposition}
    
\begin{proof}
    Consider the game $G = (\vec{w}; \vec{p})$ where $\vec{w} = (0.45, 0.30, 0.25)$ and $\vec{p} = (2, 1, 1))$. The only outcome that is utilitarian, and Pareto-optimal, is $([1\ |\ 2\ |\ 3], (2, 1, 1))$. However, the outcome \\$([2\ 3\ |\ 1], (1, 1, 1))$, where the last prize is unclaimed, is SCD/ MCD-stable as no player can gain more by deviating.
\end{proof}

Next, we draw a connection between the notion of wMCD stability and Pareto optimality.

\begin{proposition} \label{prop:wMCD=>Pareto}
    All wMCD-stable outcomes are Pareto-optimal. 
\end{proposition}

\begin{proof}
    Suppose an outcome is wMCD-stable but not Pareto-optimal. Then all players can deviate via a wMCD to the Pareto-improved outcome, which contradicts the assumption that the outcome is wMCD-stable.
\end{proof}

Finally, we find that a result analogous to the celebrated `Coase theorem' \cite{coase1960problem} holds in the model we consider.

\begin{proposition}
    All MCDT-stable outcomes are Pareto-optimal, even if we extend the set of feasible outcomes to allow for transfers between coalitions (i.e., replacing the requirement that for all $ C \in \pi: \sum_{i \in C} x_i = v(C, \pi)$ by $\sum_{C \in \pi} v(C, \pi) = \sum_{\{j:j \leq |\pi|\}} p_j$).
\end{proposition}

\begin{proof}
    The proof is analogous to that of Proposition~\ref{prop:wMCD=>Pareto}. Note that the space of feasible outcomes is significantly enriched by allowing inter-coalition transfers, which also extends the Pareto-frontier.
\end{proof}

\section{Computational complexity} \label{sec:complexity}

In this section, we study the computational complexity of three classes of decision problems related to the stability of a given outcome in a game, and the problem of deciding whether the core of a game is non-empty, all with respect to different deviation concepts. Henceforth, given a $k$-WVG $G$, a deviating partition $\pi_D \in \Pi(D)$ where $D \subseteq N$, and a set of target prizes $\vec{p}_t \subseteq \vec{p}$ such that $\vert \vec{p}_t \vert = \vert \pi_D \vert$, we will say that $\vec{p}_t$ is \textit{attainable by} $\pi(D)$ if for all partitions $\pi' \in \Pi(N\setminus D)$, the $i^{\text{th}}$ largest deviating coalition in $\pi_D$ wins a prize at least as valuable as the $i^{\text{th}}$ largest prize in $\vec{p}_t$. The problems are defined as follows, for $\theta \in \set{SCD, MCD, wMCD, MCDT}$:

\begin{quote}\underline{\textsc{Attainable-Prizes}}:\\
\emph{Given}: $k$-WVG $G = (N,\vec{w},\vec{p}, \succ)$, deviating partition $\pi_D \in \Pi(D)$ where $D \subseteq N$, and target prizes $\vec{p}_t \subseteq \vec{p}$ where $|\vec{p}_t| = |\pi(D)|$.\\
\emph{Question}: Is $\vec{p}_t$ attainable by $\pi_D$?
\end{quote}

\begin{quote}\underline{\textsc{$\exists$-$\theta$}}:\\
\emph{Given}: $k$-WVG $G = (N,\vec{w},\vec{p}, \succ)$, outcome $\omega = (\pi,\vec{x})$.\\
\emph{Question}: Is $\omega$ susceptible to a $\theta$?
\end{quote}

\begin{quote}
\underline{\textsc{Non-Empty-$\theta$-Core}}:\\
\emph{Given}: $k$-WVG $G = (N,\vec{w},\vec{p}, \succ)$.\\
\emph{Question}: Is the $\theta$-core of $G$ non-empty?
\end{quote}

\vspace*{1ex}\noindent\textbf{Complexity upper bounds:} We begin by studying the \textsc{Attainable-Prizes} problem, which we will use as an oracle in the \textsc{$\exists$-$\theta$} problems.

\begin{lemma} \label{lem:attainable_prizes}
    \textsc{Attainable-Prizes} is in co-NP.
\end{lemma}

\begin{proof}
    We can guess a partition $\pi_{N \setminus D}$ of the residual players $N \setminus D$ and verify that it prevents the deviating coalition from winning all the target prizes in polynomial time: rank all the coalitions from both the deviating partition and the residual partition, and check that the residual coalitions win one of the target prizes. This is in NP and corresponds to the complement of \textsc{Attainable-Prizes}. Therefore, \textsc{Attainable-Prizes} is in co-NP.
\end{proof}

With Lemma \ref{lem:attainable_prizes}, we now derive the upper bounds on the complexity of the $\exists$-$\theta$ set of decision problems.

\begin{theorem}\label{thm:membership_complexity}
    For $\theta \in \set{SCD,MCD,wMCD,MCDT}$, \textsc{$\exists$-$\theta$} is in~$\Sigma^P_2$.
\end{theorem}

\begin{proof}
    We show that checking for the existence of a $\theta$ is in NP, given access to an oracle for \textsc{Attainable-Prizes}, which is in co-NP by Lemma \ref{lem:attainable_prizes}.
    
    We guess a deviating partition $\pi_D$ and target prizes $\vec{p}_t$ such that the relevant incentive requirements for $\theta$ are satisfied, i.e., that if $\pi_D$ attains $\vec{p}_t$,then it satisfies the definition of $\theta$. Note that for SCD, the target prize is the smallest prize which exceeds $x(D)$; for MCD, the target prizes are the smallest prizes that exceed the payoffs to each deviating coalition; for wMCD, they are the smallest prizes that equal or exceed the payoffs of each deviating coalition, with one strict inequality required; for MCDT, we only require that the sum of target prizes exceeds the total payoffs of all deviating players. We can verify that the target prizes satisfy these requirements in polynomial time.
    
    Then, we call the \textsc{Attainable-Prizes} oracle with the game $G$, deviating partition $\pi_D$, and target prizes $\vec{p}_t$ as the input. The answer is ``yes'' if $\pi_D$ is a $\theta$. Therefore, \textsc{$\exists$-$\theta$} is in $\text{NP}^{\text{co-NP}} = \Sigma^p_2$.
\end{proof}

Given Theorem \ref{thm:membership_complexity}, upper bounds on the computational complexity of deciding whether a given outcome is $\theta$-stable and \textsc{Non-Empty-$\theta$-Core} naturally follow. Observing that checking for $\theta$-stability is the complement of \textsc{$\exists$-$\theta$}, we have the following:

\begin{corollary}
    For $\theta \in \set{SCD,MCD,wMCD,MCDT}$, a game $G$, and an outcome $\omega$, deciding whether $\omega$ is $\theta$-stable is in $\Pi^P_2$.
\end{corollary}
    
Using an oracle to decide whether an outcome is $\theta$-stable, we can readily derive an upper bound on the complexity of deciding whether the $\theta$-core of a game is empty:

\begin{corollary}
    For $\theta \in \set{SCD,MCD,wMCD,MCDT}$, \textsc{Non-Empty-$\theta$-Core} is in $\Sigma^P_3$.
\end{corollary}

While we believe the above upper bounds are tight (e.g., $\exists$-SCD is $\Sigma^P_2$-complete), it is an open question whether there are matching complexity lower bounds. In the remainder of this section, we provide a lower bound for the problem $\exists$-SCD.

\vspace*{1ex}\noindent\textbf{Complexity lower bounds:} In this section, we derive lower bounds on the computational complexity of the decision problems for the existence of beneficial deviations.

\begin{theorem}\label{thm:exists_scd_complexity}
    \textsc{$\exists$-SCD} is $NP$-hard.
\end{theorem}

\begin{proof}
    We reduce from the \textsc{Partition} problem, which is known to be $NP$-complete \cite{johnson1979computers}. An instance of \textsc{Partition} is given by a finite set $A = \set{s_1,\ldots,s_n}$ of positive integers, and the question is to determine whether there exists a subset $A' \subseteq A$ such that $\sum_{i \in A'} s_i = \sum_{i \in A \setminus A'} s_i$.
    
    Suppose that we are given an instance $A = \set{s_1,\ldots,s_n}$ of \textsc{Partition} and assume without loss of generality that the elements of $A$ are ordered, so that $s_1 \geq s_2 \geq \ldots \geq s_n$. Let $S = \frac{1}{2}\cdot\sum_{i \in A} s_i$. Now we can assume that $s_i < S$ for all $i \in \set{1,\ldots,n}$ because if $s_i > S$ for some $i$, then the answer to \textsc{Partition} would trivially be ``no'', and if $s_i = S$ for some $i$, then it is easy to find the answer to \textsc{Partition} is trivially ``yes'' by definition of $S$. With this, we construct a $k$-WVG $G = (\vec{w};\vec{p})$ with $n+2$ players $N = \set{1,2,\ldots,n+2}$, weights given by $\vec{w} = \set{s_1,s_2,\ldots,s_n,S,S}$, and prizes given by $\vec{p} = (2S+\varepsilon,S,s_1,\ldots,s_n)$, for some $\varepsilon < 1$.
    
    Now, consider the outcome $\omega = (\pi,\vec{x})$, where $\pi = [n+2\ \vert\ n+1\ \vert\ 1\ \vert\ 2\ \vert \ldots \vert\ n]$ and $\vec{x} = (s_1,\ldots,s_n, S, 2S+\varepsilon)$. We assume that the tie-breaking order can be represented succinctly and is defined so that 1) $\set{n+2} \succ \set{n+1}$, 2) $\set{n+1,n+2} \succ \set{1,\ldots,n}$, and 3) for any coalitions $C_1,C_2 \in \Pi(N)$ such that $w(C_1) = w(C_2)$, if player $1$ is in one of these coalitions, then ties are broken in favour of the coalition containing them. We show that the answer to \textsc{Partition} is ``yes'' if and only if there exists an SCD from $\omega$.
    
    For the forward direction, suppose that $A$ is a ``yes'' instance of \textsc{Partition} and let $G = (\vec{w};\vec{p})$ be the $k$-WVG constructed as described above. Because $A$ is a ``yes'' instance of \textsc{Partition}, there exists a subset $A' \subsetneq A$ such that $\sum_{i \in A'} s_i = S$. Now, either $A'$ or $A \setminus A'$ contains $s_1$. Let $C$ be the coalition of players in $N$ with weights corresponding to either of $A'$ or $A \setminus A'$ that contains player $1$ (who has weight $s_1$). We will show that $C \cup \set{n+1}$ forms an SCD from $\omega$. To see why, note first that $w(C \cup \set{n+1}) = 2S$ and $x(C \cup \set{n+1}) = 2S$. Thus, the only prize that $C \cup \set{n+1}$ will benefit from winning is the first prize, whose value is $2S+\varepsilon$. Secondly, observe that for the remaining players $R = N \setminus (C \cup \set{n+1})$, we have $w(R) = 2S$. Since ties are broken in favour of the coalition containing player $1$, i.e., coalition $C$, it follows that even if all of the remaining players in $R$ join a single coalition to prevent $C \cup \set{n+1}$ from winning first prize, ties will be broken in favour of this deviating coalition and hence, they will be able to successfully deviate to win first prize and improve their payoff. Thus, there exists an SCD from $\omega$.
    
    For the reverse direction, suppose that $A$ is a ``no'' instance of \textsc{Partition} and consider the same construction $G = (\vec{w};\vec{p})$ and outcome $\omega$. We will show that $\omega$ is an SCD-stable outcome by considering all possible SCDs by different coalitions, which can be grouped into three categories depending on which prize they aim to win: first prize, second prize, or any other prize.
    
    Firstly, we consider deviations that aim to secure first prize. Note that player $n+2$ will never want to deviate from $\omega$, because they already win the highest amount that they could possibly attain. Now, by definition of SCDs, any coalition $C$ wanting to win first prize must satisfy $x(C) < 2S+\varepsilon$. Moreover, any coalition $C$ that can guarantee the first prize must also satisfy $w(C) \geq 2S$. Otherwise, the remaining players could prevent them from winning the first prize by all grouping together. Now, it cannot be the case that $C = \set{1,\ldots,n}$ forms an SCD, because it is assumed that $\set{n+1,n+2} \succ \set{1,\ldots,n}$. Thus, if such a deviating coalition $C$ were to exist, it must consist of player $n+1$ and some subset $B \subsetneq \set{1,\ldots,n}$. However, because $A$ is a ``no'' instance of \textsc{Partition}, it must be the case that either $w(B) > S$ or $w(B) < S$. If $w(B) > S$, then by construction of $\omega$ and the fact that the elements in $A$ are all positive integers, we have $x(B \cup \set{n+1}) \geq 2S + 1 > 2S + \varepsilon$ and hence, $B \cup \set{n+1}$ has no incentive to deviate from $\omega$. On the other hand, if $w(B) < S$, then $w(B \cup \set{n+1}) < 2S$, so this coalition does not have sufficient weight to secure first prize. Thus, there is no SCD for the first prize.
    
    Secondly, we consider deviations that aim to secure second prize. Now, such deviations must only include players from $\set{1,\ldots,n}$, as players $n+1$ and $n+2$ already achieve at least the value of the second prize. In this case, a similar argument holds as for the first case. If a deviating coalition $C$ exists, then it must satisfy $x(C) < S$ and $w(C) > S$, otherwise they could not beat player $n+1$ who wins second prize in $\omega$. However, because $\omega$ was defined so that $w_i = x_i$ for all $i \in \set{1,\ldots,n}$, it cannot be the case that both of these conditions are satisfied. Hence, there is no SCD for the second~prize.
    
    Finally, we consider the case where a subset $C \subsetneq  \set{1,\ldots,n}$ may want to deviate to win any prize $p \in \set{p_3,\ldots,p_n}$. Again, we see that any such coalition must satisfy $x(C) < p$ and $w(C) > w_j$, where $w_j$ is the player who wins prize $p$ in $\omega$. However, because the considered players' weights are equal to their payoffs under $\omega$ the second condition is equivalent to $x(C) > p$ and hence, both conditions cannot be satisfied by any coalition $C \subsetneq  \set{1,\ldots,n}$. Thus, there is no SCD for any prize.
\end{proof}

Using a similar construction, we obtain the same lower bound on the complexity of $\exists$-MCD and $\exists$-wMCD. We omit the proof here as the line of reasoning is similar to that of Theorem \ref{thm:exists_scd_complexity}.

\begin{restatable}{theorem}{existsMCDwMCD}
    \textsc{$\exists$-MCD} and \textsc{$\exists$-wMCD} are NP-hard.
\end{restatable}

\section{Related Work} \label{sec:related_work}
The related models of standard WVGs and PFGs, of which $k$-WVGs are a subset, have been the subject of extensive study. We discuss some of the related literature here.

\vspace*{1ex}\noindent\textbf{Complexity of standard WVGs:}
The stability of standard WVGs has been well-studied (see, for example, \cite{chalkiadakis2011computational}). It is shown in \cite{elkind2008coalitionstructuresinWVG} that deciding if the core of a WVG is non-empty is \textit{NP-hard} and in $\Sigma^P_2$ and the related problem of deciding whether an outcome is in the core is co-NP-complete. In \cite{grecoMPS11}, it was shown that the upper bound can be lowered to $\Delta^P_2$.

\vspace*{1ex}\noindent\textbf{Characteristic function games and partition function games:} As we noted, the standard WVG is a class of characteristic function games (see for example \cite{elkind2008coalitionstructuresinWVG}), while $k$-WVGs are a class of PFGs (see for example \cite{koczy2018partitionfunctiongames,thrall1962generalized,thrall1963n,zhao1992hybrid}), where each partition of players can invoke a different characteristic function for the coalitions within that partition. Such games can also be represented using embedded MC-nets \cite{michalak2010logic}, which have the benefit of being fully expressive and exponentially more concise than the conventional PFG representation in some cases. Our model is always exponentially more concise than the general representation, but focuses on a natural and intuitive subset of all PFGs.

\vspace*{1ex}\noindent\textbf{Games with multiple prizes and ranking:}
One class of characteristic function games that involves multiple prizes are known as \textit{Threshold Task Games} (TTGs) \cite{gal2020threshold}. In TTGs, there are several tasks, each with an associated threshold and value, such that the value of a coalition is given by the highest-valued task whose associated threshold is no greater than the coalition's weight. As with standard WVGs, there are no externalities to coalition formation in TTGs, which are an essential feature of $k$-WVGs. 

The idea of having the outcomes of a game depend on a \emph{ranking} of its participants, be they players or coalitions, was also proposed in a \emph{non-cooperative} setting by Brandt et al. in \cite{brandt_etal:2009a}. However, in notable distinction to our work, the outcomes of their so-called \emph{Ranking Games} are rankings over players rather than over coalitions of players. Moreover, they focused on Nash equilibrium, correlated equilibrium, and the price of anarchy, rather than on more cooperative stability concepts like the core and its variants.

\vspace*{1ex}\noindent\textbf{Assumptions about the residual players:}
Many different concepts of stability exist in PFGs. One dimension in which they differ is how the deviating players expect the residual players to respond: the $\alpha$-core \cite{AumannPeleg1960} assumes that the residual players would act against the deviating coalition, such that a deviation is only feasible if the deviating coalition is better off under all partitions of the residual players; the $\omega$-core \cite{Shenoy1979} assumes that the residual players would act favourably toward the deviating coalition, such that a deviation is feasible if the deviating coalition is better off under some partition of the residual players; the $s$-core \cite{Chander1997} assumes that the residual players would form singleton coalitions; the $m$-core \cite{Maskin2003, Hafalir2007, Ambec2008, McQuillin2009} assumes that the residual players form a single coalition; and the $\delta$-core \cite{Hart1983} assumes that the residual players do not react and remain in their respective coalitions.
Our concept of the SCD-core for the $k$-WVG coincides with the $\alpha$-core of the more general PFGs. 

More nuanced assumptions about the residual players can also be made where various degrees of incentive compatibility are involved. Interested readers should see \cite{koczy2018partitionfunctiongames} for a survey of stability concepts.

\vspace*{1ex}\noindent\textbf{Deviations involving multiple coalitions:} Another dimension in which deviation concepts differ in PFGs is whether the deviating players form a single coalition or partition into multiple coalitions. In \cite{krus2009CostAllocation, bloch2014}, deviating players can form a partition if doing so increases the total value of all deviating coalitions. This corresponds to the concept of MCDT in this paper, where we implicitly allow transfers \emph{between} coalitions to satisfy the incentive compatibility requirements of the deviating coalitions.

\vspace*{1ex}\noindent\textbf{Efficiency} It has been noted in \cite{koczy2018partitionfunctiongames} that, in partition form games, an outcome that is stable to some deviation concept by a single coalition may not be Pareto-efficient. The simple intuition here is that the deviating coalition may gain a greater collective value if they instead form a partition of multiple deviating coalitions. \cite{Koczy2007RecursiveCore} removes this inefficiency by considering outcomes that are stable to deviations involving multiple coalitions. In Section \ref{sec:efficiency} we showed that this is indeed true for $k$-WVGs. However, we also showed that in $k$-WVGs, wMCD-stability is sufficient for ensuring Pareto-efficiency without resorting to inter-coalition transfers.

\section{Conclusions} \label{sec:conclusions}

$k$-WVGs provide a model for studying coalition formation where coalitions compete to win prizes based on their relative strengths. We began exploring some of its interesting properties, including stability concepts, the emptiness of the core, the computational complexity of stability, and the efficiency of stable outcomes. Further directions of research may include: finding tight complexity bounds for solution concepts, further characterise the core for games with all players and all prizes, exploring fairness notions, adopting alternative assumptions about the residual players in a deviation (e.g., dispensing with the pessimistic outlook), and whether, given a set of players, we can design prizes in order to achieve a particular set of stable outcomes. 

\balance

\begin{acks}
Lee was supported by the Oxford-Taiwan Graduate Scholarship and the Oxford-DeepMind Graduate Scholarship. Wooldridge and Harrenstein were supported by the UKRI under a Turing AI World Leading Researcher Fellowship (EP/W002949/1) awarded to Wooldridge.
\end{acks}



\bibliographystyle{ACM-Reference-Format} 
\bibliography{refs}


\clearpage

\ifarXiv
    \appendix
    \section{Supplementary Material}

\Threeplayertwoprizegame*

\begin{proof}
    Firstly, observe that there are only two tie-breaking orders consistent with $\set 1\succ\set 2\succ\set 3$, namely, one with $\set 1\succ\set{2,3}$ and one with $\set{2,3}\succ \set 1$.
	
	\paragraph{Case 1:} $\set 1\succ\set{2,3}$. Exploiting Remark~\ref{rem:hierarchy}, it suffices to show that:
	\begin{enumerate}[label=({\roman{*}})]
        \item\label{itemi} $\vec{x} = (R, 1, 0)$ is MCDT-stable. 
        \item\label{itemii} $\vec{x} = (R, 1, 0)$ is the only SCD-stable payoff vector. 
	\end{enumerate}
	
	For~\ref{itemi}, assume $\vec x=(R,1,0)$ and consider $\pi'=[1\  |\ 2\ |\ 3]$. 
	By Remark~\ref{rem:irrelevance_of_coalition_structure}, it suffices to show that $(\pi',\vec x)$ is MCDT-stable.
	To see this, observe that any deviating coalition must be either $\set{1,3}$ or $\set{2,3}$, and that either partition $\pi''=[1\ 3\ |\ 2]$ or partition $\pi'''=[2\ 3\ |\ 1]$ forms.
	Observe that $x_1+x_3=v(\set{1,3},\pi'')$, if the former, and $x_2+x_3=v(\set{2,3},\pi''')$ if the latter.
	Hence, we may conclude that $(\pi',\vec x)$ is MCDT-stable.

    For~\ref{itemii}, assume for the contrapositive that $\vec x\neq(R,1,0)$. 
	If $x_1<R$, recall that $\set 1\succ\set{2,3}$, and player~$1$ can form a singleton SCD and guarantee first prize on their own. 
	As~$\omega$ is an outcome, $x_1 \leq R$, we may assume that $x_1=R$. 
	With $\vec x\neq(R,1,0)$ it follows that $x_2<1$ and $x_3>0$ and $\pi=[1\ |\ 2\ 3]$.
	In this case, player~$2$ can SCD-deviate on their own and guarantee second prize.
	In either case,~$\omega$ fails to be SCD-stable.

	\paragraph{Case 2:} $\set{2,3}\succ \set 1$ and $1 < R \leq 2$. Exploiting Remark~\ref{rem:hierarchy}, it suffices to show that:
	\begin{enumerate}[label=({\roman{*}})]
		\item\label{itemiv} $\vec{x} = (1, 1, R-1)$ is MCD-stable. 
		\item\label{itemv} $\vec{x} = (1, 1, R-1)$ is not wMCD-stable. 
        \item\label{itemvi} $\vec{x} = (1, 1, R-1)$ is the only SCD-stable payoff. 
	\end{enumerate}

	For~\ref{itemiv}, first consider deviations of a single deviating coalition. No single player can win $p_1$ on its own, and therefore any deviation to win $p_1$ must consist of at least 2 players. However, any two players already have a collective payoff greater than or equal to $p_1$. Neither player 1 nor~2 would deviate (on their own or in coalition with others) to win $p_2$, since they already have payoff equating to $p_2$. Player 3 is not capable to winning $p_2$ on their own, since the remaining players can form the MCD $[1\ |\ 2]$ to take both prizes. Next, consider deviations involving all players (whether in two or three deviating coalitions): It is clear that there is no incentive for a deviation involving all players, as players have a collective payoff equally all the prizes combined. Finally, consider deviations with two deviating singleton coalitions. Such deviations cannot include player 3, as player 3 would receive no prize in such deviations (player 1 and 2 would claim the first and second prize respectively). They also cannot involve $[1\ |\ 2]$, as player 2 would win the second prize, which is not a strict improvement on its existing payoff. Thus there is no feasible MCD.
	
    For~\ref{itemv}, see that a wMCD of $[1\ |\ 2]$ guarantees payoffs of $R$ for player 1 and $1$ for player 2, which is a strict improvement for player 1.

	For~\ref{itemvi}, let $i, j \in \{1, 2\}$ and $i \neq j$, and consider any alternative payoff with $x_i \neq 1$. If $x_i > 1$, then $x_j + x_3 < R$, and therefore $\{j,3\}$ would form an SCD to win $p_1$. If $x_i < 1$, then $i$ can form a singleton SCD to win $p_2$. Thus, an SCD-stable payoff must have $x_1 = x_2 = 1$ and $x_3 = R-1$.

    \paragraph{Case 3:} $\set{2,3}\succ \set 1$ and $R > 2$. Exploiting Remark~\ref{rem:hierarchy}, it suffices to show that no payoff vector $\vec{x}$ is SCD-stable. 

    Suppose there exists an SCD-stable outcome with payoff $\vec{x}$ and let $a$, $b$, $c$ be such that $x_a\ge x_b\ge x_c$. Note that $x_b + x_c < R$: This is trivially true if the grand coalition is formed. Otherwise both prizes are claimed, with total value $p_1+p_2 > 3$, which implies $x_1 > 1$ and $x_b + x_c < R$. Players $\{b, c\}$ can thus form an SCD to take $p_1$. 
\end{proof}

\existsMCDwMCD*

\begin{proof}
    We reduce from the \textsc{Partition} problem, which is known to be $NP$-complete \cite{johnson1979computers}. An instance of \textsc{Partition} is given by a finite set $A = \set{s_1,\ldots,s_n}$ of positive integers, and the question is to determine whether there exists a subset $A' \subseteq A$ such that $\sum_{i \in A'} s_i = \sum_{i \in A \setminus A'} s_i$.
    For ease of description, we assume that $s_1 \geq s_2 \geq \ldots \geq s_n$. Moreover, we will use a version of \textsc{Partition} with the following promised properties, which remains NP-complete.
    \begin{itemize}
        \item $s_1 < S := \frac{1}{2}\cdot\sum_{i \in A} s_i$. This is because if $s_i > S$ for some $i$, then the answer to the partition problem would trivially be ``no'', and if $s_i = S$ for some $i$, then it is easy to find the answer to the partition problem is trivially ``yes'' by definition of $S$. 
        
        \item If the instance is a ``yes'' instance, then for every subset such that $\sum_{i \in A'} s_i = S$, either $\{s_1,s_2\} \subseteq A'$ or $\{s_1,s_2\} \cap A' = \emptyset$.
        Indeed, we can convert every \textsc{Partition} instance into one satisfying this property by adding five additional integers, with two of the integers having a value of $6S$, and the remaining three having a value of $4S$.
        The new instance is a ``yes'' instance if and only if the original one is a ``yes'' instance. Meanwhile, the two largest integers with value $6S$ must appear in the same set in order to partition the integers into two equal sets.
    \end{itemize}

    We begin by reducing this special instance of \textsc{Partition} to \textsc{$\exists$-MCD}. For this, we construct a $k$-prize WVG $G = (\vec{w};\vec{p})$ with $n+1$ players $N = \set{1,2,\ldots,n+1}$, weights $\vec{w} = \set{s_1,s_2,\ldots,s_n,S}$ and prizes $\vec{p} = (S+\varepsilon,s_1',s_2',\ldots,s_n')$, where $s_i'= s_i + (n-i)\delta$ for all $i \in \set{1,\ldots,n}$, for some $\varepsilon < 1$, and some $\delta < \frac{\epsilon}{n^2}$. Note that under this choice of $\delta$, we have $\sum_{i=1}^n (n-i)\delta < n^2 \delta < \varepsilon$ and $p_i > p_j$ for all $1 \leq i < j \leq n+1$. Intuitively, we construct a game where $n$ of the players have weights corresponding to the values of the elements in the given set $A$ and an additional player with weight $S = \frac{1}{2}\cdot\sum_{i \in A} s_i$.
    
    Now, consider the outcome 
    \[
        \omega = ([n+1\ \vert\ 1\ \vert\ 2\ \vert \ldots \vert\ n],(s_1',\ldots,s_n',S+\varepsilon)),
    \]
    where each player is in a singleton coalition and prizes are simply allocated according to the weight of each player. 
    We assume that for any coalitions $C_1,C_2 \in \Pi(N)$ such that $w(C_1) = w(C_2)$, if both players $1$ and $2$ are in one of these coalitions, then ties are broken in favour of the coalition containing them. Additionally, we assume that ties are always broken in favour of singletons and that any tie-breaks not involving player $1$ or singletons follow an arbitrary rule that can be implemented efficiently.
    We show that the answer to the partition problem is ``yes'' if and only if there exists an MCD/wMCD from $\omega$.
    
    \paragraph{Forward Direction}
    For the forward direction, suppose that $A$ is a ``yes'' instance of the partition problem and let $G = (\vec{w};\vec{p})$ be the $k$-prize WVG constructed as described above. We argue that there exists an MCD. (The existence of a wMCD then follows immediately because $MCD(G) \subset wMCD(G)$ for all games $G$.)
    Since $A$ is a ``yes'' instance of \textsc{Partition}, there exists a subset $A' \subset A$ such that $\sum_{i \in A'} s_i = S$. 
    Now, according to the promised property, either $A'$ or $A \setminus A'$ contains both $s_1$ and $s_2$. Let $C$ be the coalition of players in $N$ with weights corresponding to either of $A'$ or $A \setminus A'$ that contains players $1$ and $2$ (who have weights $s_1$ and $s_2$, respectively), and let $C' = \{1,\dots,n\}\setminus C$. We will show that $\{C\} \cup \{\{i\}: i \in C'\}$ forms an MCD from $\omega$. Indeed, we have $w(C) = W(n+1) = S > s_i$ for all $i \in C'$, and because tie-breaks are assumed to be in favour of coalitions containing both players 1 and 2, it follows that $C$ will win the first prize. Additionally, $C$ will gain a strict improvement, because \[\sum_{i \in C} s_i' = \sum_{i \in C} s_i + (n-i)\delta = S + \sum_{i \in C} (n-i)\delta < S + \varepsilon,\] by definition of $\delta$. 
    Player $n+1$ will win the largest of the remaining prizes, which has value $s'_1$. 
    Let $C' = \{\ell_1, \dots, \ell_\lambda\}$ such that $s_{\ell_1} \ge s_{\ell_2} \ge \dots \ge s_{\ell_\lambda}$. Because all prizes are guaranteed to have different values by the addition of different multiples of $\delta$, player $\ell_1$ will win the prize of value $s'_2$ and gain a strict improvement. The original prize of player $\ell_1$ can be passed on to player $\ell_2$, and likewise every player $\ell_i \in C'$ will win a prize with value at least $s'_{\ell_{i-1}}$ and gain a strict improvement since $s'_{\ell_{i-1}} > s'_{\ell_i}$. 
    Thus, there exists an MCD from $\omega$.
    
    \paragraph{Reverse Direction}
    For the reverse direction, suppose that $A$ is a ``no'' instance of the partition problem and consider the same construction $G = (\vec{w};\vec{p})$ and outcome $\omega = (\pi, \vec{x})$. We will show that $\omega$ is a stable outcome against all possible wMCDs.
    (Similarly, stability against MCD then follows immediately since $MCD(G) \subset wMCD(G)$ for all games $G$.)
    
    Suppose for the sake of contradiction that there exists a wMCD $\pi_D$ from $\omega$ for some $D \subseteq N$, and suppose that $C \in \pi_D$. 
    Firstly, we consider the case where $C$ secures the first prize $p_1$ of value $S+\varepsilon$. If $C = \set{S}$, then this coalition does not gain anything. Moreover, we can assume without loss of generality that $n+1 \notin D$, because player $n+1$ will never stand to gain anything by deviating, and in all scenarios where this player would consider joining $D$, they will have no effect on the prizes that other coalitions in $D$ win. Otherwise, we have $C \neq \set{S}$. For the deviation to be beneficial for $C$, we must have that $x(C) < S+\varepsilon$. However, since $A$ is a ``no'' instance of \textsc{Partition} by assumption, this must mean that either $w(C) < S$, or $w(C) = S$ and $C$ does not contain both players $1$ and $2$. In the first case, we must have $w(C) < S = w(n+1)$. This means that player $n+1$ can prevent $C$ from winning the first prize, which contradicts the assumption. In the second case, because ties are broken in favour of singletons and $C$ is assumed not to contain both players 1 and 2, it follows that $S \succ C$ and hence, $C$ cannot win the first prize.
    
    Next, we consider the case where every coalition $C \in \pi_D$ wins one of the prizes in $\vec{p} \setminus \set{p_1}$. Let $C$ be the largest coalition in $\pi_D$. Then, there are two cases: either $\vert C \vert > 1$, or $\vert C \vert = 1$.
    
    If $\vert C \vert > 1$, then let $s_i'$ be the value of the prize that $\vert C \vert$ would like to claim. If $w(C) > w_i$, then $x(C) > x_i = s_i'$ since all weights are integers and hence, $C$'s payoff is strictly worse by winning prize $i$. Therefore, we can assume that $w(C) \leq w_i$. Because $C$ is assumed to the largest coalition in $D$ and ties are broken in favour of singletons, any player with weight at least $s_i$ must not be part of $D$. Thus, in order to win prize $p_i$, it must be that $w(C) > s_i$, otherwise all players in $\set{1,\ldots,i}$ can prevent $C$ from winning the $i$-th prize by remaining as singletons. This is a contradiction with the assumption that $w(C) \leq w_i = s_i$ and hence, either $C$ cannot win prize $p_i$, or it would be strictly worse off if it did. In this case, we can conclude that $\pi_D$ does not form a wMCD from $\omega$.
    
    If $\vert C \vert = 1$, then let $i$ be the only member of $C$. Because $\set{i}$ is the largest coalition in $D$, we can infer that no member of $C$ who is larger than $i$ is part of $D$ and will hence prevent $C$ from winning a higher prize than they currently obtain under $\omega$. In this case, we can assume that $i$ wins the same prize $p_i$ as they did under $\omega$ and move on to the second largest coalition $C_2$ in $\pi_D$. Repeating the process, if $\vert C_2 \vert > 1$, then the members of $C_2$ will not be able to benefit from any prize that they can secure. If $\vert C_2 \vert = 1$, then we move on to the third largest coalition, and so on. Eventually, we will either reach a coalition of size larger than $1$, in which case we can conclude that the deviation is not beneficial, or we reach the last coalition in $\pi_D$ which has only one member. In this case, all coalitions in $\pi_D$ would consist of singletons and would not improve over what they already achieved in $\omega$. Hence, there is no wMCD.
\end{proof}

\fi

\end{document}